\documentclass[10pt,journal]{IEEEtran} 
\usepackage[T1]{fontenc}
\usepackage{amsmath,amssymb,amsfonts}
\usepackage{algorithm}
\usepackage{algpseudocode}
\usepackage{graphicx}
\usepackage{textcomp}
\usepackage{stmaryrd}
\usepackage{braket}
\usepackage{multirow}
\usepackage{mathtools}
\usepackage{hyphenat}
\usepackage{array}

\usepackage[square,numbers,sort&compress]{natbib}
\newcommand{\vect}[1]{\mathbf{#1}} 

\begin{document}

\title{Encoder Circuit Optimization for Non-Binary Quantum Error Correction Codes in Prime Dimensions: An Algorithmic Framework}

\author{Aditya~Sodhani,~\IEEEmembership{Graduate Student Member,~IEEE},
        and~Keshab~K.~Parhi,~\IEEEmembership{Life Fellow,~IEEE}%
\thanks{A. Sodhani is with the University of Minnesota Twin Cities, Minneapolis, MN 55455 USA (e-mail: sodha005@umn.edu).}%
\thanks{K. K. Parhi is with the University of Minnesota Twin Cities, Minneapolis, MN 55455 USA (e-mail: parhi@umn.edu).}%
\thanks{This work was supported in part by the National Science Foundation under Grant CCF-2243053.}}

\maketitle

\begingroup\renewcommand\thefootnote{}\footnotetext{\emph{Preprint.}
This work has been submitted to the IEEE for possible publication in \emph{IEEE Transactions on Quantum Engineering}. Copyright may be transferred without notice, after which this version may no longer be accessible.}\addtocounter{footnote}{-1}\endgroup

\begin{abstract}
Quantum computers are a revolutionary class of computational platforms with applications in combinatorial and global optimization, machine learning, and other domains involving computationally hard problems. While these machines typically operate on qubits—quantum information elements that can occupy superpositions of the basis $|0\rangle$ and $|1\rangle$ states—recent advances have demonstrated the practical implementation of higher-dimensional quantum systems (qudits) across various hardware platforms. In these hardware realizations, the higher-order states are less stable and thus remain coherent for a shorter duration than the basis $|0\rangle$ and $|1\rangle$ states. Moreover, formal methods for designing efficient encoder circuits for these systems remain underexplored. This limitation motivates the development of efficient circuit techniques for qudit systems (d-level quantum systems). Previous works have typically established generating gate sets for higher-dimensional codes by generalizing the methods used for qubits. In this work, we introduce a systematic framework for optimizing encoder circuits for prime-dimension stabilizer codes. This framework is based on novel generating gate sets whose elements map directly to efficient Clifford gate sequences. We demonstrate the effectiveness of this method on key codes, achieving a 13--44\% reduction in encoder circuit gate count for the qutrit (d = 3) $[[9,5,3]]_3$, $[[5,1,3]]_3$, and $[[7,1,3]]_3$ codes, and a 9--21\% reduction for the ququint (d = 5) $[[10,6,3]]_5$ code when compared to prior work. We also achieved circuit depth reductions upto 42\%.
\end{abstract}


\maketitle

\section{Introduction}
\label{sec:introduction}
Quantum Computing is a fast growing field of computation that has shown great potential to revolutionize problem solving using principles of quantum mechanics and thus achieve exponential speedups over its classical counterpart. Despite their many advantages, quantum computers are highly prone to errors and decoherence, which require robust error corrections \cite{reed2012realization}. The present Quantum Computing approach is primarily dominated by qubits, which operate on two distinct energy levels. However, there has been interest in its qudits (multi-level quantum systems) \cite{baker2020efficient} counterpart, especially qutrit, because of the various advantages it offers \cite{gokhale2019asymptotic, gokhale2020extending}. It is possible to achieve qudit-based quantum computing using photonic systems \cite{gao2020computer}, continuous spin systems \cite{bartlett2002quantum}, ion traps \cite{klimov2003qutrit}, nuclear magnetic resonance \cite{dogra2014determining}, and molecular magnets \cite{leuenberger2001quantum}. In addition, qudit systems expand the search space, which strengthens security protocols \cite{hajji2021qutrit} and increases channel capacity in the field of quantum cryptography \cite{bruss2002optimal,blok2021quantum} and communication. Qudit systems expand the state space and allow simultaneous control operations, which reduces circuit complexity, simplifies experimental setups, and improves algorithm efficiency \cite{wang2020qudits}. For example, recent researches \cite{ralph2007efficient,nikolaeva2022decomposing,yeh2022constructing} have shown that using qutrits (3-level quantum system) helps to reduce the circuit complexity while implementing important and complex quantum gates, such as the Toffoli and T gates.

Non-binary quantum systems enhance quantum error correction (QEC) by reducing code sizes, leading to effective encoding schemes to achieve higher error thresholds and fault tolerance which are necessary for reliable error correction. In the stabilizer formalism introduced by Gottesman \cite{gottesman1997stabilizer}, quantum codes are constructed by identifying a set of commuting operators whose joint eigenspaces serve as code spaces; the associated eigenvalues, or syndromes, facilitate the detection and correction of errors. For non-binary codes, Ashikhmin and Knill \cite{ashikhmin2001nonbinary} and later Ketkar et al. \cite{ketkar2006nonbinary} extended this framework to qudits with dimension \(d = p^k\), providing encoding procedures that involve projecting the all-zero state onto the code space and applying normalizer operations to generate the different codewords. Grassl \textit{et al.} \cite{grassl2003efficient} proposed an encoding method for non-binary quantum stabilizer codes by conjugating the stabilizers using Clifford operations, thereby transforming them into a canonical form. Nadkarni \textit{et al.} \cite{nadkarni2021encoding} generalized Grassl's procedure from codes based on classical linear codes to the much broader class of classical additive codes.

Although qudit-based quantum systems are promising, they present a significant practical challenge in hardware realizations, as higher-order energy levels are comparatively less stable \cite{sharma2024compilation}. These excited states are more susceptible to environmental noise and, therefore, remain coherent for a shorter period of time \cite{ogunkoya2024qutrit} as compared to lower states. Therefore, methods must be developed to operate efficiently within the small, coherent window. Thus, reducing circuit complexity, specifically by minimizing total gate counts and achieving lower circuit depths, becomes essential \cite{mondal2024optimization, mondal2024quantum} for implementing reliable qudit-based algorithms. This is the primary motivation of our research. Although existing works mentioned earlier provide procedures for designing encoding QECC circuits, they primarily utilize direct generalizations of qubit-based gates \cite{sawicki2017universality, sabo2024trellis}. These approaches, however, do not necessarily target an optimal generating gate set tailored to qudit systems. This motivates the need for a systematic framework for optimizing circuits, particularly encoder circuits for qudit codes, by developing gate sets that provide efficient quantum circuits in terms of the number of gates and circuit depth. 

The contributions of this paper are three-fold: 1)\ We develop an algorithmic framework to derive gate sets that yield an optimal gate count and reduced circuit depth for encoder design, improving upon previous gate sets used for general d-level quantum systems. 2)\ Using this framework, we identify the optimal gate sets for d=3 (qutrit) and d=5 (ququint) systems. We also provide a formal proof for the operators corresponding to the gates for the d=3 case.
3)\ We systematically derive encoder circuits for various quantum error-correcting codes for d=3 and d=5 systems. We also illustrate the optimality of the proposed approach.

This paper is organized as follows. In Section II, we discuss concepts of qudit-based quantum systems, including state representation, the generalized Pauli group, the stabilizer formalism, and Clifford operations. In Section III, we discuss the general encoding procedure. In Section IV, we discuss our proposed algorithmic framework for finding optimal gate sets. In Section V, we use the framework to find optimal generating gate sets for various non-binary quantum codes. In Section VI, we provide the operator derivations corresponding to the gate sets derived from the proposed algorithmic framework. Finally, conclusions are provided in Section VII.

\section{Quantum Systems over qudits}

\subsection{State Representation}
A qudit is the quantum counterpart of a classical d-ary system. Its state is represented by a normalized vector in a d-dimensional Hilbert space, which is a complex vector space equipped with an inner product. To illustrate, consider the most basic example: a 2-ary classical system, in which one digit is represented by either 0 or 1. Then the state of its quantum counterpart (qubit or qudit for d = 2) is represented as state vector with orthonormal basis vectors as \(|0\rangle\) and \(|1\rangle\). Thus, the qubit state  is given by: 

\begin{equation}
|\phi\rangle = \alpha |0\rangle + \beta |1\rangle
\end{equation}
where \(\alpha, \beta \in \mathbb{C}\) and the normalization condition \(|\alpha|^2 + |\beta|^2 = 1\) must hold true.

More generally, the state of a qudit is a vector in the complex vector space spanned by a set of d orthonormal basis vectors, which are labeled by either \{\(|0\rangle, |1\rangle, \dots, |d-1\rangle\)\} or by the elements of a finite field. Thus, the general state of a qudit $|\psi\rangle$ is given by : 
\begin{equation}
|\psi\rangle = \alpha_0 |0\rangle + \alpha_1 |1\rangle + \cdots + \alpha_{d-1} |d-1\rangle
\end{equation}
where the normalization condition is \(\sum_{i=0}^{d-1} |\alpha_i|^2 = 1\).

\subsection{Qudit Pauli Group}
For qubits, the Pauli group comprises the identity operator $I$ and the Pauli matrices $\sigma_x$, $\sigma_y$, and $\sigma_z$ with phase factors {±1,±i}:

\[
\sigma_x = 
\begin{pmatrix}
0 & 1 \\
1 & 0
\end{pmatrix}, \quad
\sigma_y = 
\begin{pmatrix}
0 & -i \\
i & 0
\end{pmatrix}, \quad
\sigma_z = 
\begin{pmatrix}
1 & 0 \\
0 & -1
\end{pmatrix}.
\]
While generalizing to qudits, we consider systems where every qudit is a state in d-dimensional Hilbert space where d =  \( p^m \), i.e., dimension d is power of prime p. 
With this, we define the generalized Pauli operators $X(a)$ and $Z(b)$ for a single qudit:
\[
X(a) := \sum_{x \in \mathbb{F}_d} |x + a\rangle\langle x|, \quad Z(b) := \sum_{z \in \mathbb{F}_d} \omega^{\text{tr}(b z)} |z\rangle\langle z|
\]
where $a, b \in \mathbb{F}_d$, \(\text{tr}(\cdot)\) denotes the trace operation from \( \mathbb{F}_d \) to \( \mathbb{F}_p \) which is defined as:
\begin{equation}
\text{tr}(x) = \sum_{i=0}^{m-1} x^{p^i} .
\end{equation}

For an \( n \)-qudit system, a Pauli product is defined as: \( \omega^c X(\mathbf{a}) Z(\mathbf{b}) \). Here \( X(\mathbf{a}) = X(a_1) \otimes \dots \otimes X(a_n) \) and \( Z(\mathbf{b}) = Z(b_1) \otimes \dots \otimes Z(b_n) \), where \(\mathbf{a}, \mathbf{b} \in \mathbb{F}_d^n\), i.e., \( \mathbf{a} = (a_1, \dots, a_n) \), \( \mathbf{b} = (b_1, \dots, b_n) \), \( a_i, b_i \in \mathbb{F}_d \) , and \(c \in \mathbb{F}_p\).

Thus, the \(d\)-dimensional Pauli group \(\mathcal{P}_{n}^{(d)}\) is defined by 
all possible Pauli products over \(n\)-qudits. For instance, consider the Pauli group of a \textbf{single qutrit} ($n=1, ~d=3$). Since $d=3$ is prime, $p=3$ and $m=1$, so the trace function is simply the identity, $\text{tr}(x) = x$. The operators have the form
\[
  \omega^{\lambda}\,X(i)\,Z(j),
  \quad
  \text{where}
  \quad
  \lambda,\,i,\,j \in \mathbb{F}_{3}
  \quad\text{and}\quad
  \omega \;=\; e^{2\pi i/3}.
\]
Here, each of the nine unphased operators \(X(i)\,Z(j)\) (with \(i,j \in \{0,1,2\}\)) 
can be multiplied by any of the three overall phases \(\omega^\lambda\). 
This yields a total of
\[
  3~(\text{phases})
  \;\times\;
  3~(\text{values of }i)
  \;\times\;
  3~(\text{values of }j)
  \;=\;
  27
\]
elements in the single‐qutrit Pauli group. Specifically, the nine unphased 
operators are given by:
\[
  \bigl\{
    I,\;
    Z(1),\;
    Z(2),\;
    X(1),\;
    X(1)\,Z(1),\;
\]
\[
    X(1)\,Z(2),\;
    X(2),\;
    X(2)\,Z(1),\;
    X(2)\,Z(2)
  \bigr\}.
\]
Each of these nine operators appears with each of the three phases $\omega^\lambda$ (where $\lambda \in \{0, 1, 2\}$) to form the full group.

\subsection{Nice Error Basis and Error Group}
For a \textbf{single qudit} system ($n=1$), consider the following set of $d^2$ unitary operators, which form the basis for single-qudit errors:
\begin{equation}
    \mathcal{E} = \{ X(a) Z(b) : a, b \in \mathbb{F}_d \}.
\end{equation}

It can be shown \cite{ashikhmin2001nonbinary} that these \( d^2 \) operators form an orthogonal operator basis with respect to the inner product $\langle A, B \rangle = \text{tr}(A^\dagger B)$. This set is called a "nice error basis" and has the following properties:
\begin{enumerate}
    \item It contains the identity operator (for $a=b=0$).
    \item It is orthogonal: $\text{tr}(E_1^\dagger E_2) = d \cdot \delta_{E_1, E_2}$ for all $E_1, E_2 \in \mathcal{E}$.
    \item The product of any two basis elements is, up to a phase, another basis element: for any $E_1, E_2 \in \mathcal{E}$, their product $E_1 E_2 = c E_3$ for some $E_3 \in \mathcal{E}$ and a phase factor $c$.
    \item The operators satisfy the commutation relation: $Z(b)X(a) = \omega^{\text{tr}(ab)} X(a)Z(b)$.
    \item They follow the product rule: $X(a)Z(b) \cdot X(a')Z(b') = \omega^{\text{tr}(ba')} X(a+a')Z(b+b')$.
\end{enumerate}
Similar to the qubit Pauli group, we can define the n-qudit \emph{error group} $G_n$ as the set of all Pauli products including phase factors:
\begin{equation}
G_n = \left\{ \omega^c X(\mathbf{a}) Z(\mathbf{b}) \;\middle|\; \mathbf{a}, \mathbf{b} \in \mathbb{F}_d^n, \; c \in \mathbb{F}_p \right\} 
\end{equation}

For example, two errors $E_1 = \omega^{c_1} X(\mathbf{a_1}) Z(\mathbf{b_1})$ and $E_2 = \omega^{c_2} X(\mathbf{a_2}) Z(\mathbf{b_2})$ in $G_n$ commute if and only if their trace symplectic product vanishes, i.e.,
\begin{equation}
    \text{tr}(\mathbf{a_1} \mathbf{b_2} - \mathbf{a_2} \mathbf{b_1}) = 0.
\end{equation}
Using the product rule from Property (5), it is straightforward to verify the product of these two errors:
\begin{align}
    E_1 E_2 &= \omega^{c_1+c_2} \omega^{\text{tr}(\mathbf{b_1} \mathbf{a_2})} X(\mathbf{a_1} + \mathbf{a_2}) Z(\mathbf{b_1} + \mathbf{b_2}) \\
    E_2 E_1 &= \omega^{c_1+c_2} \omega^{\text{tr}(\mathbf{b_2} \mathbf{a_1})} X(\mathbf{a_1} + \mathbf{a_2}) Z(\mathbf{b_1} + \mathbf{b_2})
\end{align}
Clearly, $E_1 E_2 = E_2 E_1$ only when $\omega^{\text{tr}(\mathbf{b_1} \mathbf{a_2})} = \omega^{\text{tr}(\mathbf{b_2} \mathbf{a_1})}$, which is equivalent to the condition $\text{tr}(\mathbf{a_1} \mathbf{b_2} - \mathbf{a_2}\mathbf{b_1}) = 0 \pmod p$.

\subsection{Stabilizer Codes}
Let $S$ be an Abelian subgroup of the error group $G_n$ such that $-I \notin S$. 
We call $S$ the \textbf{stabilizer group}.

An $[[n,k]]$ non-binary stabilizer code, $C$, is the $d^k$-dimensional 
subspace of the $n$-qudit Hilbert space, $\mathcal{H}_d^n$, that is 
stabilized by every element of $S$. In other words, the code $C$ is the 
simultaneous +1 eigenspace of all operators in $S$. This is expressed as:
\begin{equation}
    C = \left\{ |\psi\rangle \in \mathcal{H}_d^n \;\middle|\; s|\psi\rangle = |\psi\rangle 
    \text{ for all } s \in S \right\}.
\end{equation}

\subsection{Symplectic Inner Product}
The n-qudit Pauli group $\mathcal{P}_n^{(d)}$ has a natural classical representation. 
The phase-space vector corresponding to an operator $X(\mathbf{a})Z(\mathbf{b})$ 
is the $2n$-dimensional vector $(\mathbf{a}, \mathbf{b})$ over $\mathbb{F}_d$. 
The set of all such vectors forms a \textbf{symplectic module}, which is a 
$2n$-dimensional module over the ring $\mathbb{Z}_d$.

This module is equipped with a \textbf{symplectic inner product (SIP)}, a 
non-degenerate, skew-symmetric bilinear form defined as:
\begin{equation}
    \langle (\mathbf{a}_1, \mathbf{b}_1), (\mathbf{a}_2, \mathbf{b}_2) \rangle 
    := \text{tr}(\mathbf{a}_1 \cdot \mathbf{b}_2 - \mathbf{a}_2 \cdot \mathbf{b}_1)
\end{equation}
where the trace is taken from $\mathbb{F}_d$ to its prime subfield $\mathbb{F}_p$.

The SIP's importance is that it directly determines the commutation relations 
of the Pauli operators. Two operators, $E_1 = X(\mathbf{a}_1)Z(\mathbf{b}_1)$ and 
$E_2 = X(\mathbf{a}_2)Z(\mathbf{b}_2)$, have the following commutation relation:
\begin{equation}
    E_1 E_2 = \omega^{\langle (\mathbf{a}_1, \mathbf{b}_1), (\mathbf{a}_2, \mathbf{b}_2) \rangle} E_2 E_1
\end{equation}
The operators commute if and only if their symplectic inner product is zero.

\subsection{Clifford Operators and the Symplectic Form}
Consider a single qudit of dimension $d$ with Pauli group, $\mathcal{P}_1$, 
generated by the $X$ and $Z$ operators, which are of order $d$. As discussed 
previously, any element of the Pauli group (up to overall phases) can be 
written as $X(a)Z(b)$.

A unitary operator $U$ belongs to the \textbf{Clifford group}, $\mathcal{C}$, 
if it normalizes the Pauli group; that is, it maps elements of the Pauli 
group to other elements under conjugation:
\[
    U(X(a)Z(b))U^{-1} \in \mathcal{P}_1.
\]
In the classical representation for an $n$-qudit system, these operators 
correspond to $2n \times 2n$ matrices over $\mathbb{Z}_d$ that preserve the 
symplectic inner product (SIP). A key tool in this description is the 
\textbf{standard symplectic form}, given by the block matrix:
\begin{equation}
    S =
    \begin{pmatrix}
        0 & I_n \\
        -I_n & 0
    \end{pmatrix},
\end{equation}
where $I_n$ is the $n \times n$ identity matrix. A matrix $N$ is said to be 
\textbf{symplectic} if it satisfies:
\begin{equation}
    N^T S N = S.
\end{equation}
This condition ensures that the SIP is preserved under $N$. For $N$ to correspond 
to a valid Clifford operator, it must also satisfy the \textbf{determinant condition}:
\begin{equation}
    \det(N) \equiv 1 \pmod d.
\end{equation}
By encoding the commutation relations of quantum operators into these 
symplectic matrices, one obtains a concise and powerful way to analyze 
Clifford gates.

For example, consider the unitary $U$ that performs the following transformation:
\[
    U X(a) U^{-1} = Z(-a), \quad U Z(b) U^{-1} = X(b).
\]
This operator is in the Clifford group, $\mathcal{C}$. Its action on the 
phase-space vector $(a, b)$ is given by the matrix:
\begin{equation}
    F := \begin{pmatrix}
        0 & 1 \\
        -1 & 0
    \end{pmatrix}.
\end{equation}
To verify that $F$ represents a valid Clifford operator, we must check that 
it satisfies \textbf{both} conditions. We see that $\det(F) = (0)(0) - (1)(-1) = 1$, 
which satisfies the determinant condition. We must also separately verify that 
it is symplectic:
\begin{align}
    F^T S F 
        &= \begin{pmatrix} 0 & -1 \\ 1 & 0 \end{pmatrix} 
           \begin{pmatrix} 0 & 1 \\ -1 & 0 \end{pmatrix} 
           \begin{pmatrix} 0 & 1 \\ -1 & 0 \end{pmatrix} \nonumber \\
        &= \begin{pmatrix} 1 & 0 \\ 0 & 1 \end{pmatrix} 
           \begin{pmatrix} 0 & 1 \\ -1 & 0 \end{pmatrix} 
         = S.
\end{align}
Since both conditions are met, $F$ is a valid representation of a Clifford operator.

\subsection{Known Clifford Operators and Their Symplectic Transformations}
\subsubsection{Fourier Transform Operator (DFT)}
The DFT over a finite field \(\mathbb{F}_d\) is given by:

\begin{equation}
    \text{DFT}_d := \frac{1}{\sqrt{d}} \sum_{x,z \in \mathbb{F}_d} 
    \omega^{\operatorname{tr}(xz)} |z\rangle \langle x|
\end{equation}

where:
\begin{itemize}
    \item \( \mathbb{F}_d \) is the finite field for the qudit of dimension \( d = p^m \).
    \item \( \omega = e^{2\pi i/p} \) is a primitive \( p \)-th root of unity, where $p$ is the characteristic of the field.
    \item \( \operatorname{tr}(\cdot) \) is the field trace from \(\mathbb{F}_d\) to its prime subfield \(\mathbb{F}_p\).
\end{itemize}

For the specific case of qutrits, where \( d=3 \), the DFT can be simplified. Since $\mathbb{F}_3$ is a prime field ($p=3, m=1$), the trace function is the identity ($\operatorname{tr}(x) = x$). The DFT matrix is defined as:
\begin{equation}
    \text{DFT}_3 := \frac{1}{\sqrt{3}} \sum_{x,z \in \mathbb{F}_3} \omega^{xz} |z\rangle \langle x|
\end{equation}
where the sum is over elements \( x, z \in \{0,1,2\} \) and \( \omega = e^{2\pi i /3} \). The matrix form of \( \text{DFT}_3 \) in the computational basis \( \{ |0\rangle, |1\rangle, |2\rangle \} \) is given by:
\begin{equation}
    \text{DFT}_3 =
    \frac{1}{\sqrt{3}}
    \begin{bmatrix}
        1 & 1 & 1 \\
        1 & \omega & \omega^2 \\
        1 & \omega^2 & \omega
    \end{bmatrix}
\end{equation}

As discussed previously, the action of the DFT on the phase-space vector $(a, b)$ that represents a Pauli operator is a linear transformation. This symplectic representation of the DFT is given by:
\begin{equation}
    \overline{\text{DFT}} = 
    \begin{bmatrix}
        0 & -1 \\
        1 & 0
    \end{bmatrix}.
\end{equation}

\subsubsection{Multiplication Operator}
For a general qudit system of dimension $d$, and any nonzero element 
$\gamma \in \mathbb{F}_d^*$, the \emph{multiplication operator} is defined as:
\begin{equation}
    M_\gamma := \sum_{y \in \mathbb{F}_d} |\gamma y\rangle \langle y|
\end{equation}

When $d=3$, we work in the prime field $\mathbb{F}_3 = \{0,1,2\}$. For 
$\gamma \in \{1,2\}$, the operator $M_\gamma$ acts on the qutrit basis 
$\{|0\rangle, |1\rangle, |2\rangle\}$ by permuting the basis vectors. 
In the computational basis, $M_\gamma$ is a $3 \times 3$ permutation matrix 
whose $(z,y)$-th entry is given by:
\begin{equation}
    (M_\gamma)_{z,y} =
    \begin{cases}
        1, & \text{if } z \equiv \gamma y \pmod{3},\\
        0, & \text{otherwise}.
    \end{cases}
\end{equation}
Concretely:
\begin{equation}
    M_1 =
    \begin{pmatrix}
        1 & 0 & 0\\
        0 & 1 & 0\\
        0 & 0 & 1
    \end{pmatrix}
    \quad \text{and} \quad
    M_2 =
    \begin{pmatrix}
        1 & 0 & 0\\
        0 & 0 & 1\\
        0 & 1 & 0
    \end{pmatrix}.
\end{equation}
Here, $M_1$ is the identity, and $M_2$ swaps $|1\rangle \leftrightarrow |2\rangle$ 
while fixing $|0\rangle$.

The symplectic (phase-space) representation of $M_\gamma$ is a $2 \times 2$ 
matrix that describes its action on the phase-space vector $(a,b)$. 
The transformation is:
\begin{equation}
    M_\gamma \longmapsto
    \overline{M}_\gamma =
    \begin{pmatrix}
        \gamma^{-1} & 0\\
        0 & \gamma
    \end{pmatrix}.
\end{equation}
Hence, $\gamma^{-1}$ scales the X-coordinate ($a$), and $\gamma$ scales the 
Z-coordinate ($b$). For qutrits, when $\gamma=2$, its inverse $\gamma^{-1}$ 
is also 2, since $2 \cdot 2 \equiv 1 \pmod{3}$. Thus, the symplectic matrix 
for $M_2$ is:
\[
    \overline{M}_2 = \begin{pmatrix} 2 & 0\\ 0 & 2 \end{pmatrix}.
\]

\subsubsection{$P_\gamma$ (Quadratic Phase) Operator}

Let $d$ be an odd prime and $\omega := e^{2\pi i/d}$. For any $\gamma \in \mathbb{F}_d$, define
\begin{equation}
    P_\gamma \;=\; \sum_{y\in\mathbb{F}_d} \omega^{-\frac{1}{2}\gamma y^2}\,\ket{y}\!\bra{y},
\end{equation}
where $\tfrac{1}{2}$ denotes the multiplicative inverse of $2$ in $\mathbb{F}_d$; all arithmetic is modulo $d$. (For $d=2$, a different definition is required and is not used here.)

Ignoring overall phases, the conjugation relations are given by:
\begin{align}
    P_\gamma^{-1} X_\alpha P_\gamma &= \omega^{\frac{1}{2}\gamma \alpha^2}\, X_\alpha Z_{\gamma\alpha},\\
    P_\gamma^{-1} Z_\beta P_\gamma  &= Z_\beta,
\end{align}
for $\alpha,\beta\in\mathbb{F}_d$.

Using phase–space \emph{row} vectors with right multiplication, $(a,b)\in\mathbb{F}_d^2$ update is represented by:
\begin{equation}
    (a,b)\ \longmapsto\ (a,\ b+\gamma a),
\end{equation}
which is represented by the symplectic matrix
\begin{equation}
    \overline{P}_\gamma \;=\;
    \begin{pmatrix}
        1 & \gamma \\
        0 & 1
    \end{pmatrix},
    \qquad (a,b)\,\overline{P}_\gamma=(a,\ b+\gamma a).
\end{equation}

For the qutrit case $d=3$ and $\gamma=1$,
\[
    P_1=\mathrm{diag}\!\bigl(1,\ \omega^{-\frac{1}{2}},\ \omega^{-2}\bigr),
    \qquad
    \overline{P}_1=\begin{pmatrix}1&1\\0&1\end{pmatrix},
\]
so $(a,b)\mapsto(a,b+a)$ is consistent with the above relations.

\subsubsection{Addition Operator (Two-Qudit Operator)}

Consider two qudits, each of dimension $d$, with computational bases 
$\{|x\rangle_1 : x \in \mathbb{F}_d\}$ and $\{|y\rangle_2 : y \in \mathbb{F}_d\}$. 
The \textbf{ADD operator} is defined by its action on the basis states:
\begin{equation}
    \mathrm{ADD}^{(1,2)} |x\rangle_1 |y\rangle_2 = |x\rangle_1 |x + y\rangle_2,
\end{equation}
with all arithmetic modulo $d$.
In the qutrit case ($d=3$), this operator is a $9 \times 9$ permutation matrix. 
For example, the basis state $|1,2\rangle$ is mapped to 
$|1, 1+2 \pmod 3\rangle = |1,0\rangle$.

The action of this gate on the phase-space vector $(a_1, a_2, b_1, b_2)$ is also 
well-defined. Ignoring overall phases (standard in symplectic representations), the ADD operator (equivalent to a CNOT gate) under right-multiplication corresponds to the  following linear transformation:
\[
    (a_1, a_2, b_1, b_2) \longmapsto (a_1, a_2 + a_1, b_1 - b_2, b_2)
\]
This transformation is represented by the following $4 \times 4$ symplectic matrix 
(assuming row vectors and right-multiplication):
\begin{equation}
    \overline{\mathrm{ADD}}^{(1,2)} =
    \begin{pmatrix}
        1 & 1 & 0 & 0 \\
        0 & 1 & 0 & 0 \\
        0 & 0 & 1 & 0 \\
        0 & 0 & -1 & 1
    \end{pmatrix}.
\end{equation}
In effect, this `ADD` operation adds the X-component of the first qudit to the 
X-component of the second, and it subtracts the Z-component of the second qudit 
from the Z-component of the first.

\subsubsection{SWAP Operator (Two-Qudit Operator)}

The \textbf{SWAP} gate exchanges two $d$-level systems:
\begin{equation}
    \mathrm{SWAP}^{(1,2)} \ket{x}_1\ket{y}_2 = \ket{y}_1\ket{x}_2,
\end{equation}
and can be written as
\begin{equation}
    \mathrm{SWAP}_d \;=\; \sum_{x,y\in\mathbb{F}_p} 
    \ket{x}_d\!\bra{y}_d \,\otimes\, \ket{y}_d\!\bra{x}_d .
\end{equation}

With the convention of phase-space \emph{row} vectors ordered as $(a_1,a_2,b_1,b_2)$ and right-multiplication (all arithmetic modulo $d$), SWAP induces
\begin{equation}
    (a_1,a_2,b_1,b_2)\ \longmapsto\ (a_2,\,a_1,\,b_2,\,b_1),
\end{equation}
i.e., it swaps both the $X$- and $Z$-components of qudits~1 and~2.  
Equivalently, its $4\times 4$ symplectic matrix is given by:
\begin{equation}
    \overline{\mathrm{SWAP}}^{(1,2)} \;=\;
    \begin{pmatrix}
        0 & 1 & 0 & 0 \\
        1 & 0 & 0 & 0 \\
        0 & 0 & 0 & 1 \\
        0 & 0 & 1 & 0
    \end{pmatrix},
\end{equation}
so that
\[
(a_1,a_2,b_1,b_2)\,\overline{\mathrm{SWAP}}^{(1,2)}
= (a_2,a_1,b_2,b_1)\, .
\]

\section{Encoding Algorithm}
It is necessary to know that the single-qudit gates (mentioned in section II) $\{\mathrm{DFT},\,M_{\gamma},\,P_{\gamma}\}$ form a generating set for the (single-qudit) Clifford group in prime dimension $d$~\cite{grassl2003efficient}.
We will modify Grassl's decoding procedure \cite{grassl2003efficient} to accommodate the encoding of quantum codes over prime dimension d. This has been proposed by Nadkarni \textit{et al.} \cite{nadkarni2021encoding}.
For each stabilizer generator $S_i = X(\mathbf{a}_i)Z(\mathbf{b}_i)$, we define its 
classical representation as the vector $\mathbf{h}_i = (\mathbf{a}_i | \mathbf{b}_i)$. 
The check matrix $\mathcal{H}_{(X|Z)}$ is then constructed from these vectors as its rows. Given these, the encoding steps are: 

\textit{Step 1:} For the first row of the check matrix $\mathcal{H}_{(X|Z)}$, using the operators and their symplectic transformations defined in the generating gate set, transform the Pauli operators (except $(0|0)$) over each qudits to $(1|0)$. Let this operation be $T_{1}$. Apply $T_{1}$ to all the rows below the current row (here, row 1). If the operation over the first qudit is $(0|0)$ and $k$ is the first qudit with operation $(1|0)$ after applying $T_{1}$, then preform SWAP(1,$k$)

\textit{Step 2:} Now, for each qudit $j$ whose corresponding pair of entries in the first row 
of $\mathcal{H}_{(X|Z)}$ is $(1|0)$, apply the column operations equivalent to an 
$\mathrm{ADD}^{(1,j)}$ gate to the entire matrix below the current row (here, row 1). This procedure transforms the first row of $\mathcal{H}_{(X|Z)}$ into the vector $(1,0\dots, 0, 0, 0 | 0, 0, \dots, 0,0)$. Let this operation be $A_{1}$.

\textit{Step 3:} After Step (1) and (2), the combined operator we have is $T_{1}$$A_{1}$. Repeat steps (a) and (b) for all the rows of the matrix $\mathcal{H}_{(X|Z)}$. At the end the combined operator we get is $T_1 A_1 T_2 A_2 \dots T_{n-k} A_{n-k}$.

\textit{Step 4: }
The final operation is to apply $DFT^{-1}$ to all qudits corresponding to pivot qudits in the 
transformed parity check matrix. Let this operator be $F^{-1}$. Thus, the complete encoding operator is thus given 
by the product $T_1 A_1 T_2 A_2 \dots T_{n-k} A_{n-k}F^{-1}$.

We now apply this encoding procedure to five qutrit code $\llbracket 5,1,3\rrbracket_3$. The operators for the symplectic transformations are tabulated in TABLE~\ref{tab:qutrit_transformations_scalebox}. 

\begin{table}[ht]
\centering
\caption{Table of qutrit operations (Section III) and their symplectic transformations}
\scalebox{1.4}{%
    \begin{tabular}{|c|c|}
    \hline
    \textbf{Qutrit Operations} & \textbf{Transformations} \\
    \hline
    $M_{2}DFT$       & $(0\,,2)\;\to\;(1\,,0)$ \\
    \hline
    $P_{1}$ & $(1\,,2)\;\to\;(1\,,0)$ \\
    \hline
    $P_{1}M_{2}$           & $(2\,,1)\;\to\;(1\,,0)$ \\
    \hline
    $P_{2}M_{2}$           & $(2\,,2)\;\to\;(1\,,0)$ \\
    \hline
    $M_{2}$          & $(2\,,0)\;\to\;(1\,,0)$ \\
    \hline
    $DFT$       & $(0\,,1)\;\to\;(1\,,0)$ \\
    \hline
    $P_{2}$       & $(1\,,1)\;\to\;(1\,,0)$ \\
    \hline
    \end{tabular}
}

\label{tab:qutrit_transformations_scalebox}
\end{table}
The check matrix \cite{anwar2012qutrit} $\mathcal{H}_{(X|Z)}$ for the code is given by:
\[\mathcal{H}_{(X|Z)} \;=\;
\left[\begin{array}{ccccc|ccccc}
1 & 0 & 0 & 2 & 0 & 0 & 1 & 2 & 0 & 0\\
0 & 1 & 0 & 0 & 2 & 0 & 0 & 1 & 2 & 0\\
2 & 0 & 1 & 0 & 0 & 0 & 0 & 0 & 1 & 2\\
0 & 2 & 0 & 1 & 0 & 2 & 0 & 0 & 0 & 1
\end{array}\right]
\]

\textit{Step 1:} In the first step, we transform each pair $(\alpha_i,\beta_i)$ of the first row that is nonzero to $(1,0)$ using the transformations summarized in Table 1. This is achieved by the transformation:
\begin{equation}
T_{1} \;=\;
\mathrm{id} \;\otimes\;
\mathrm{DFT} \;\otimes\;
 M_2\mathrm{DFT}\;\otimes\;
M_{2} \;\otimes\;
\mathrm{id}
\end{equation}
Moreover, this transformation has to be applied on all the rows below the 1st row.
The resulting stabilizer matrix is

\[
\mathcal{H}_{(X|Z)} = 
\left[ \begin{array}{ccccc|ccccc}
    1 & 1 & 1 & 1 & 0 & 0 & 0 & 0 & 0 & 0 \\
    0 & 0 & 2 & 0 & 2 & 0 & 2 & 0 & 1 & 0 \\
    2 & 0 & 0 & 0 & 0 & 0 & 0 & 1 & 2 & 2 \\
    0 & 0 & 0 & 2 & 0 & 2 & 1 & 0 & 0 & 1
\end{array} \right]
\]
\textit{Step 2:} The first non-zero column is the first one, so we apply $\mathrm{ADD}^{(1,j)}$ operation on all the qudits whose operator in the first row of check matrix is $(1,0)$. This gives us the transformation:
\begin{equation}
A_{1} \;:=\;
\mathrm{ADD}^{(1,2)}\;
\mathrm{ADD}^{(1,3)}\;
\mathrm{ADD}^{(1,4)}\;
\end{equation}
\textit{Step 3:} Now, we also apply this for rest of the rows of the check matrix as well. We get:

\[\mathcal{H}_{(X|Z)} \;=\;
\left[\begin{array}{ccccc|ccccc}
1 & 0 & 0 & 0 & 0 & 0 & 0 & 0 & 0 & 0\\
0 & 0 & 2 & 0 & 2 & 0 & 2 & 0 & 1 & 0\\
2 & 1 & 1 & 1 & 0 & 0 & 0 & 1 & 2 & 2\\
0 & 0 & 0 & 2 & 0 & 0 & 1 & 0 & 0 & 1
\end{array}\right]
\]

We repeat Steps (1) and (2) for rest of the rows. After applying these two steps for 2nd row we get the following operators:

\begin{equation}
T_{2} \;=\;
\mathrm{id} \;\otimes\;
M_2\mathrm{DFT} \;\otimes\;
 M_2\;\otimes\;
\mathrm{DFT} \;\otimes\;
M_2
\end{equation}

\begin{equation}
A_{2} \;:=\;
\mathrm{ADD}^{(2,3)}\;
\mathrm{ADD}^{(2,4)}\;
\mathrm{ADD}^{(2,5)}\;
\end{equation}

and the resulting check matrix is:

\[\mathcal{H}_{(X|Z)} \;=\;
\left[\begin{array}{ccccc|ccccc}
1 & 0 & 0 & 0 & 0 & 0 & 0 & 0 & 0 & 0\\
0 & 1 & 0 & 0 & 0 & 0 & 0 & 0 & 0 & 0\\
2 & 0 & 2 & 2 & 0 & 0 & 0 & 2 & 2 & 1\\
0 & 2 & 1 & 1 & 1 & 0 & 0 & 0 & 1 & 2
\end{array}\right]
\]

For the third row we get : 

\begin{equation}
T_{3} \;=\;
M_2 \;\otimes\;
\mathrm{id} \;\otimes\;
 P_2M_2\;\otimes\;
 P_2M_2\;\otimes\;
\mathrm{DFT}
\end{equation}
\begin{equation}
A_{3} \;:=\;
\mathrm{ADD}^{(3,1)}\;
\mathrm{ADD}^{(3,4)}\;
\mathrm{ADD}^{(3,5)}\;
\end{equation}

and the resulting check matrix is:

\[\mathcal{H}_{(X|Z)} \;=\;
\left[\begin{array}{ccccc|ccccc}
1 & 0 & 0 & 0 & 0 & 0 & 0 & 0 & 0 & 0\\
0 & 1 & 0 & 0 & 0 & 0 & 0 & 0 & 0 & 0\\
0 & 0 & 1 & 0 & 0 & 0 & 0 & 0 & 0 & 0\\
1 & 2 & 2 & 0 & 0 & 0 & 0 & 0 & 0 & 2
\end{array}\right]
\]

For the last rows, we have : 

\begin{equation}
T_{4} \;=\;
\mathrm{id}\;\otimes\;
M_2 \;\otimes\;
M_2 \;\otimes\;
 \mathrm{id}\;\otimes\;
M_2\mathrm{DFT}
\end{equation}
\begin{equation}
A_{4} \;:=\;
\mathrm{SWAP}^{(4,5)}\;
\mathrm{ADD}^{(4,1)}\;
\mathrm{ADD}^{(4,2)}\;
\mathrm{ADD}^{(4,3)}\;
\end{equation}

and the resulting check matrix is given by:

\[\mathcal{H}_{(X|Z)} \;=\;
\left[\begin{array}{ccccc|ccccc}
1 & 0 & 0 & 0 & 0 & 0 & 0 & 0 & 0 & 0\\
0 & 1 & 0 & 0 & 0 & 0 & 0 & 0 & 0 & 0\\
0 & 0 & 1 & 0 & 0 & 0 & 0 & 0 & 0 & 0\\
0 & 0 & 0 & 1 & 0 & 0 & 0 & 0 & 0 & 0
\end{array}\right]
\]

Thus the final encoding operator we obtain is :

\begin{equation}
E = T_1 A_1 T_2 A_2 \dots T_4 A_4F^{-1}.
\end{equation}
The resulting encoding circuit is shown in FIGURE 1. 
\begin{figure}[H]
  \centering
  \includegraphics[width=0.45\textwidth]{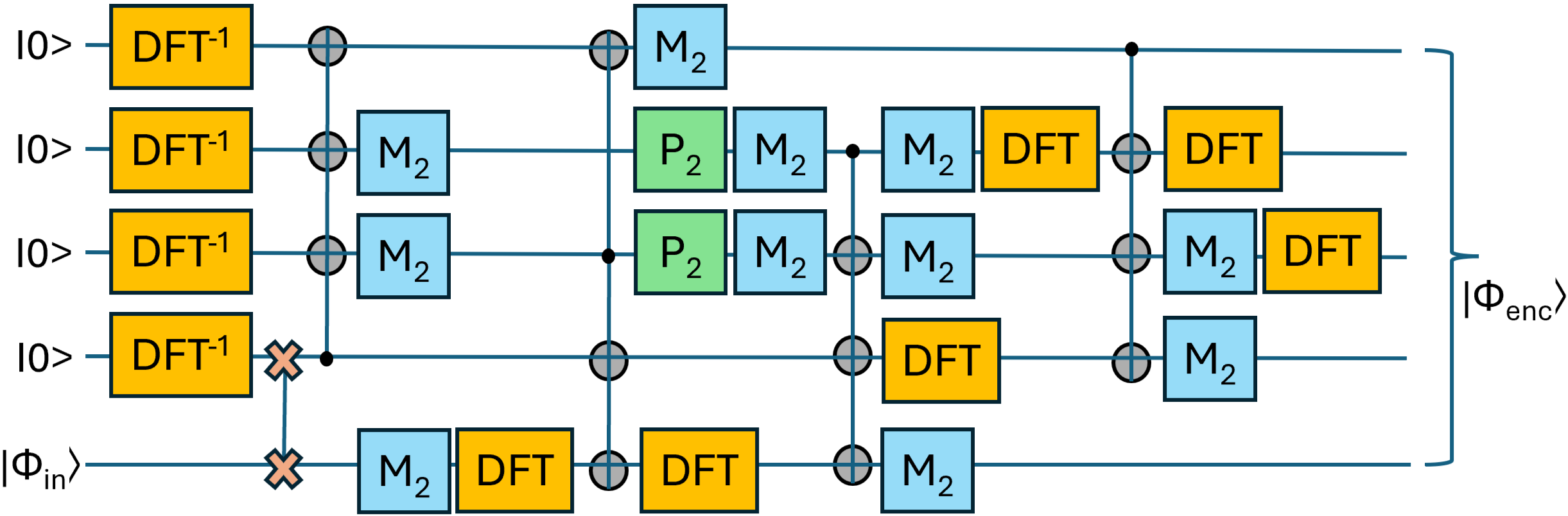}
  \caption{Encoding Circuit for Five Qutrit Code \(\llbracket 5,1,3\rrbracket_3\) using gate set of Section III.}
  \label{fig:circuit1}
\end{figure}

\section{Algorithmic Framework}
This is a multi-stage process involving pre-computation, an exhaustive search, and a final validation and scoring phase, as detailed in Algorithm~\ref{alg:find_optimal_set}.

\begin{enumerate}
    \item \textbf{Initialization and Pre-computation (Lines 4-11)} \\
    The problem space is specified at the beginning of the algorithm. Its inputs are prime dimension $d$, the desired number of operators $\text{set\_size}$, and a set of single-step constraints $\mathcal{C}$. 
    Single step constraints means the transformations that we require to be done in a single step, i.e., with a single operator (gate). The complete mathematical group of possible operators, $S_{all} = SL(2, \mathbb{F}_d)$, which represents the universe of all potential tools for the search, is then produced by the algorithm (Line 5). To optimize this search, the algorithm pre-computes a "pool" of candidate matrices, $P_i$, for each single-step constraint $(\vect{v}_i, 1) \in \mathcal{C}$ (Lines 8-10). It is guaranteed that every matrix in $P_i$ will complete the transformation from $\vect{v}_i \to \vect{v}_{target}$ in a single step.
    \item \textbf{Exhaustive Search for Candidate Sets (Lines 16-20)} \\
    This is the central component of the algorithm, which builds and assesses every possible candidate base set through systematic, exhaustive search. A sequence of nested loops makes up the structure of this search. The outermost loop (Line 13) iterates through the Cartesian product of the candidate pools, $P_1 \times \dots \times P_k$. Each tuple $(M_1, \dots, M_k)$ in this product represents a legitimate set of matrices that concurrently satisfies every single-step constraint specified in $\mathcal{C}$. For each such valid selection, an initial set of required operators, $S_{req}$, is formed, which includes the DFT and the chosen constraint-solving matrices (Line 14). After a pruning step to discard impossible combinations where the number of required unique operators already exceeds the target \texttt{set\_size} (Line 15), the algorithm proceeds to the inner loop (Line 19). This loop completes the candidate set by choosing the remaining $n_{choose}$ operators from the pool of all other available matrices, $S_{pool}$. Each complete set, denoted $S_{base}$ (Line 20), is then passed to the validation and scoring phase. This structured, nested iteration ensures that every potential solution that meets the initial constraints is methodically constructed and tested.

   \setcounter{enumi}{2} 
    \item \textbf{Validation and Pathfinding (Lines 21-23)} \\
    We apply two validation criteria for each of the constructed candidate set $S_{base}$. The set must first be a generating set for the entire group $SL(2, \mathbb{F}_d)$, which is verified by the \textsc{IsGeneratingSet} function (Line 21). If this condition is met, the \textsc{FindShortestPaths} function is then employed to find the optimal path for each required transformation (Line 23). This pathfinding is accomplished via a Breadth-First Search (BFS) that determines the shortest sequence of operations using only the operators available in the given $S_{base}$.

    \item \textbf{Scoring and Selection (Lines 24-40)} \\
    Finally, a cost metric, defined as the sum of the lengths of all shortest paths ($\text{total\_ops}$), is calculated for each valid and fully evaluated base set (Line 31). Throughout the exhaustive search, the algorithm maintains a record of the current optimal solution, consisting of the set $S_{best}$ and its associated cost, $\text{min\_ops}$. A newly evaluated set replaces the current $S_{best}$ if its total operational cost is lower. The final output of the algorithm, after all combinations have been considered, is the single set that performs all transformations with the minimum possible $\text{total\_ops}$ (Line 35).
\end{enumerate}

\begin{algorithm}
\caption{Find Optimal Symplectic Operator Set}
\label{alg:find_optimal_set}
\begin{algorithmic}[1]
\Procedure{FindOptimalSet}{$p, \text{set\_size}, \mathcal{C}$}
    \State \textbf{Input:} Prime dimension $p$, desired set size $\text{set\_size}$, constraints $\mathcal{C} = \{(\vect{v}_i, 1)\}_{i=1}^k$
    \State \textbf{Output:} An optimal operator set $S_{best}$ and their transformations, or failure.
    
    \State $\vect{v}_{target} \gets (1, 0)$
    \State $S_{all} \gets \{ M \in M_2(\mathbb{F}_d) \mid \det(M) = 1 \}$ \Comment{Generate all matrices in $SL(2, \mathbb{F}_p)$}
    \State $M_{DFT} \gets \begin{pmatrix} 0 & p-1 \\ 1 & 0 \end{pmatrix}$
    
    \State \Comment{For each constraint, find all single matrices that satisfy it}
    \For{each constraint $(\vect{v}_i, 1) \in \mathcal{C}$}
        \State $P_i \gets \{ M \in S_{all} \mid \vect{v}_i \cdot M = \vect{v}_{target} \}$
    \EndFor
    
    \State $S_{best} \gets \text{null}$; $\text{min\_ops} \gets \infty$
    
    \State \Comment{Iterate through all combinations of matrices satisfying the constraints}
    \For{each tuple $(M_1, \dots, M_k) \in P_1 \times \dots \times P_k$}
        \State $S_{req} \gets \{M_{DFT}\} \cup \{M_1, \dots, M_k\}$
        \If{$|\text{unique}(S_{req})| > \text{set\_size}$} \textbf{continue} \EndIf
        
        \State $S_{pool} \gets S_{all} \setminus \text{unique}(S_{req})$
        \State $n_{choose} \gets \text{set\_size} - |\text{unique}(S_{req})|$
        
        \For{each combination $S_{other} \subset S_{pool}$ of size $n_{choose}$}
            \State $S_{base} \gets \text{unique}(S_{req}) \cup S_{other}$
            
            \If{\textbf{not} IsGeneratingSet($S_{base}, p$)} \textbf{continue} \EndIf
            
            \State $\text{Paths} \gets \text{FindShortestPaths}(S_{base}, \vect{v}_{target})$
            
            \State \Comment{Verify constraints and calculate total operations}
            \State $\text{total\_ops} \gets 0$; $\text{valid} \gets \text{true}$
            \For{each vector $\vect{v}$ in the problem space}
                \If{$\vect{v} \notin \text{Paths}$} $\text{valid} \gets \text{false}$; \textbf{break} \EndIf
                \If{$(\vect{v}, 1) \in \mathcal{C}$ \textbf{and} $|\text{Paths}[\vect{v}]| \neq 1$} $\text{valid} \gets \text{false}$; \textbf{break} \EndIf
                \State $\text{total\_ops} \gets \text{total\_ops} + |\text{Paths}[\vect{v}]|$
            \EndFor
            
            \If{\text{valid} \textbf{and} $\text{total\_ops} < \text{min\_ops}$}
                \State $\text{min\_ops} \gets \text{total\_ops}$
                \State $S_{best} \gets S_{base}$
            \EndIf
        \EndFor
    \EndFor
    
    \State \Return $S_{best}$
\EndProcedure
\algstore{bkbreak}
\end{algorithmic}
\end{algorithm}

\begin{algorithm}
\addtocounter{algorithm}{-1} 
\caption{Find Optimal Symplectic Operator Set\\ (continued)}

\begin{algorithmic}[1]
\algrestore{bkbreak}
\Function{IsGeneratingSet}{$S, p$}
    \State \Comment{Performs a BFS on the Cayley graph of $\langle S \rangle$ to find standard generators}
    \State \Return True if $\begin{pmatrix} 1 & 1 \\ 0 & 1 \end{pmatrix}$ and $\begin{pmatrix} 0 & p-1 \\ 1 & 0 \end{pmatrix}$ are found, else False.
\EndFunction
\vspace{1em}
\Function{FindShortestPaths}{$S, \vect{v}_{target}$}
    \State \Comment{Performs a backward BFS from $\vect{v}_{target}$ using matrix inverses}
    \State \Return A map from each reachable vector to its shortest path.
\EndFunction
\end{algorithmic}
\end{algorithm}

\section{Generating Gate Set for Prime Dimension Codes}

We consider the $d=3$, i.e., the qutrit case. First, we consider the case where the number of Clifford gates in the generating set is 4, i.e., \texttt{set\_size} = 4. We set the constraints for single-step transformations to include $(0|2) \to (1|0)$, $(2|1) \to (1|0)$, and $(2|0) \to (1|0)$. Applying these as inputs to our algorithm, we get the following gates which are mentioned in Theorem 1.

\newtheorem{theorem}{Theorem} 
\newtheorem{proof}{Proof}     

\begin{theorem}
Let $J'_{\mathbb{F}_3}$ be the group of quantum operations generated by the operators $\{L, \text{DFT}, M_2, R\}$. The group of transformations induced by $J'_{\mathbb{F}_3}$ on the single-qutrit Pauli group, $\mathcal{P}_1$, is isomorphic to the special linear group $SL(2, \mathbb{F}_3)$ where the symplectic matrices corresponding to the actions of the generators are given by:
\[
    \overline{L} = \begin{pmatrix} 0 & 1 \\ 2 & 0 \end{pmatrix}, \quad 
    \overline{\text{DFT}} = \begin{pmatrix} 0 & 2 \\ 1 & 0 \end{pmatrix}, 
\]
\[
    \overline{M_2} = \begin{pmatrix} 2 & 0 \\ 0 & 2 \end{pmatrix}, \quad 
    \overline{R} = \begin{pmatrix} 0 & 2 \\ 1 & 2 \end{pmatrix}
\]
\end{theorem}
\begin{proof}
\textbf{(1) Each generator lies in $SL(2,\mathbb F_3)$}. A direct calculation gives
\begin{flalign*}
& \det\overline{L}=1, \quad \det\overline{\mathrm{DFT}}=1, \\
& \det\overline{M_2}=(2\cdot2)\equiv 1\pmod 3, \quad \det\overline{R}=1. &
\end{flalign*}
so $\overline{L},\overline{\mathrm{DFT}},\overline{M_2},\overline{R}\in SL(2,\mathbb F_3)$.

\smallskip
\textbf{(2) Extract the standard generators from the given set.}
It is known \cite{stein2007modular} that $SL(2,\mathbb F_d)$ for an odd prime $d$ is generated by the two matrices
\[
S=\begin{pmatrix}0&-1\\[2pt]1&0\end{pmatrix},\qquad
T(1)=\begin{pmatrix}1&1\\[2pt]0&1\end{pmatrix}.
\]
(Over $\mathbb F_3$, $-1\equiv 2$.) We now show that $S$ and $T(1)$ lie in the subgroup $G:=\langle \overline{L},\overline{\mathrm{DFT}},\overline{M_2},\overline{R}\rangle$.

\emph{(a) $S\in G$.} Here $S=\overline{\mathrm{DFT}}=\begin{psmallmatrix}0&2\\[1pt]1&0\end{psmallmatrix}$ is one of the given generators.

\emph{(b) Produce an elementary shear.} Compute
\[
\overline{R}\,\overline{\mathrm{DFT}}
=\begin{pmatrix}0&2\\[2pt]1&2\end{pmatrix}\!
  \begin{pmatrix}0&2\\[2pt]1&0\end{pmatrix}
=\begin{pmatrix}2&0\\[2pt]2&2\end{pmatrix}
= \overline{M_2}\,\begin{pmatrix}1&0\\[2pt]1&1\end{pmatrix}.
\]
Since $\overline{M_2}^{-1}=\overline{M_2}$ in $\mathbb F_3$, we obtain the \emph{lower} shear
\[
P:=\begin{pmatrix}1&0\\[2pt]1&1\end{pmatrix}
= \overline{M_2}\,(\overline{R}\,\overline{\mathrm{DFT}})\in G.
\]

\emph{(c) Conjugate to get an \emph{upper} shear.} Conjugating $P$ by $S$ swaps upper and lower triangular unipotents:
\[
S^{-1} P\, S
= \begin{pmatrix}1&2\\[2pt]0&1\end{pmatrix}
=: T(2)\in G.
\]
Because $T(a)T(b)=T(a+b)$ for unipotent shears, we have
\[
T(1)=T(2)^2\in G.
\]
Hence both $S$ and $T(1)$ belong to $G$.

\smallskip
\textbf{(3) Generation.} The standard result for linear groups over finite fields states that $\langle S, T(1)\rangle = SL(2,\mathbb F_d)$ for prime $d$ (this is the usual $(s,t)$-presentation with $s^4=1$, $(st)^3=s^2$, together with $t$ unipotent) \cite{gorenstein2007finite}. Specializing to $d=3$ yields $\langle S,T(1)\rangle=SL(2,\mathbb F_3)$. Since $\{S,T(1)\}\subseteq G\subseteq SL(2,\mathbb F_3)$, we conclude $G=SL(2,\mathbb F_3)$.

\smallskip
\textbf{(4) Isomorphism of induced action.} The induced action of single\hyp{}qutrit Clifford operators on Pauli exponents $(a,b)$ factors through the natural homomorphism onto $SL(2,\mathbb F_3)$; modulo global phases (which act trivially on $\mathcal P_1$), this action is faithful on the exponent vectors. Since the image of $J'_{\mathbb F_3}$ equals $SL(2,\mathbb F_3)$ by Steps (1)\,(3), the induced transformation group is (canonically) isomorphic to $SL(2,\mathbb F_3)$. This completes the proof.
\end{proof}

Table 2 shows replication of Table 1 for the new set of operators.

\begin{table}[H]
\caption{Table of qutrit operations (proposed) and their symplectic transformations}
\label{tab:qutrit_transformations_scalebox}
\centering
\scalebox{1.32}{%
    \begin{tabular}{|c|c|}
    \hline
    \textbf{Qutrit Operations} & \textbf{Transformations} \\
    \hline
    $L$       & $(0\,,2)\;\to\;(1\,,0)$ \\
    \hline
    $RM_{2}$  & $(1\,,2)\;\to\;(1\,,0)$ \\
    \hline
    $R$       & $(2\,,1)\;\to\;(1\,,0)$ \\
    \hline
    DFT$R$    & $(2\,,2)\;\to\;(1\,,0)$ \\
    \hline
    $M_{2}$   & $(2\,,0)\;\to\;(1\,,0)$ \\
    \hline
    DFT       & $(0\,,1)\;\to\;(1\,,0)$ \\
    \hline
    $LR$      & $(1\,,1)\;\to\;(1\,,0)$ \\
    \hline
    \end{tabular}
}
\end{table}

FIGURE 2 shows the encoding circuit if we apply these gates for $\llbracket 5,1,3\rrbracket_3$ code. We see a reduction of single qutrit gate count by 15.8 \% and circuit depth by 21.4 \%. The operations using the novel gate set have been tabulated in Table 3. When we use these derived gates and apply the algorithm mentioned in section IV to the 9 qutrit code \(\llbracket 9,5,3\rrbracket_3\), we get the operators summarized in Table 4 and the derived encoder circuit in FIGURE 3. The encoding circuit for this code has also been proposed by Grassl \textit{et al.} \cite{grassl2003efficient}. When we compare the number of single qutrit gates, we observe a reduction of 15 \%. 

To show the generalization of the proposed algorithm we now consider the code with next higher prime dimension, i.e., with d = 5, $[[10,6,3]]_5$ code. When we run the algorithm for \texttt{set\_size} = 4, we get the following set as mentioned in Theorem 2. When we compare it with the gate set mentioned in Section III, we see 9 \% reduction in single qudit gate count for encoder circuit and a 12 \% reduction in circuit depth. The operators for encoder circuit are tabulated in Table 5.

\begin{theorem}
Let $J'_{\mathbb{F}_5}$ be the group of quantum operations generated by the operators $\{\text{DFT}, P, Q, S\}$. The group of transformations induced by $J'_{\mathbb{F}_5}$ on the single-qudit Pauli group, $\mathcal{P}_1$ (for dimension $d=5$), is isomorphic to the special linear group $SL(2, \mathbb{F}_5)$ where the symplectic matrices corresponding to the actions of the generators are given by:
\[
    \overline{\text{DFT}} = \begin{pmatrix} 0 & 4 \\ 1 & 0 \end{pmatrix}, \quad
    \overline{P} = \begin{pmatrix} 2 & 3 \\ 2 & 1 \end{pmatrix},
\]
\[
    \overline{Q} = \begin{pmatrix} 2 & 0 \\ 0 & 3 \end{pmatrix}, \quad
    \overline{S} = \begin{pmatrix} 4 & 4 \\ 0 & 4 \end{pmatrix}
\]
\end{theorem}

\begin{proof}
The proof is similar to the proof for Theorem 1.
\end{proof}

We tabulate the generating gate set derived from our code with the general generating gate set mentioned in Section II and III for d = 3 and d = 5 in Table 6. This table is based on the number of gates we considered in our generating set.

\begin{figure}[htbp]
  \centering
  \includegraphics[width=0.4\textwidth]{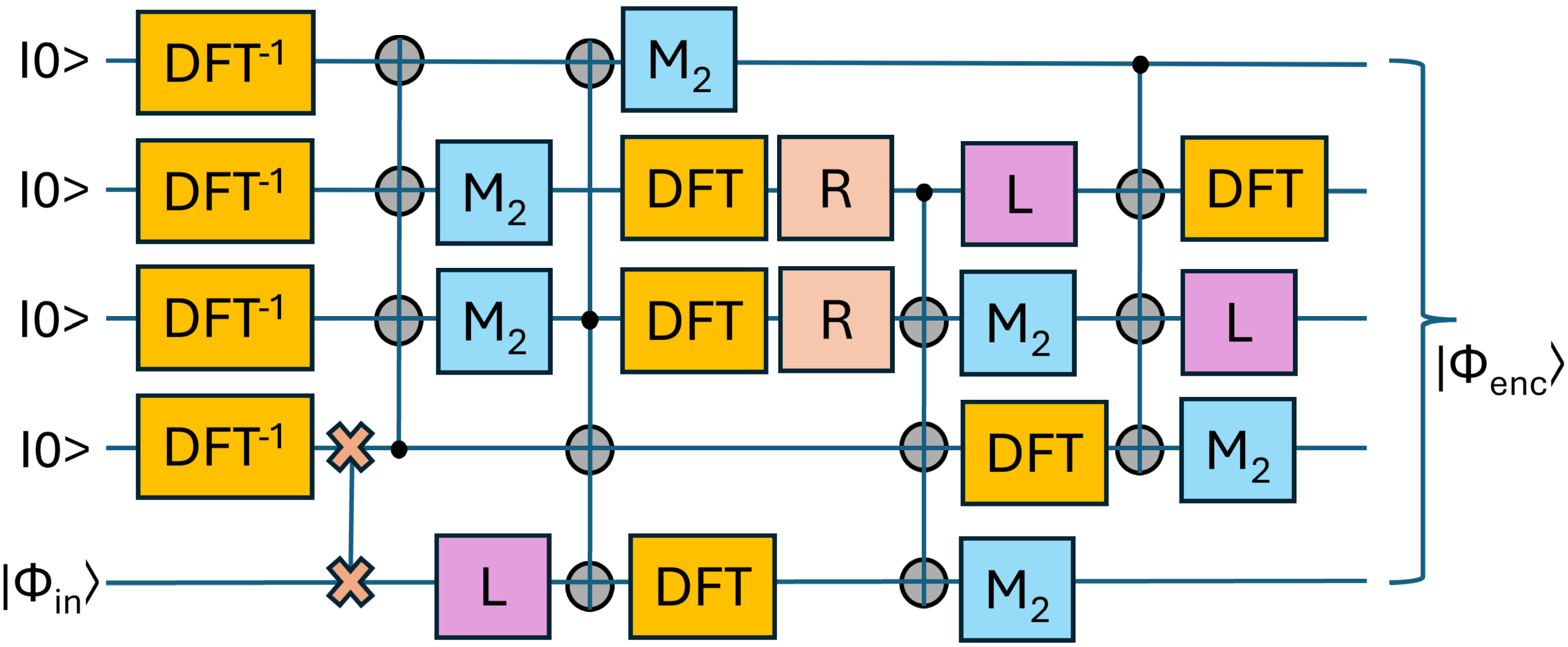}
  \caption{Encoding Circuit for Five Qutrit Code \(\llbracket 5,1,3\rrbracket_3\) using the proposed gate set.}
  \label{fig:circuit1}
\end{figure}

\begin{figure}[htbp]
  \centering
  \includegraphics[width=0.5\textwidth]{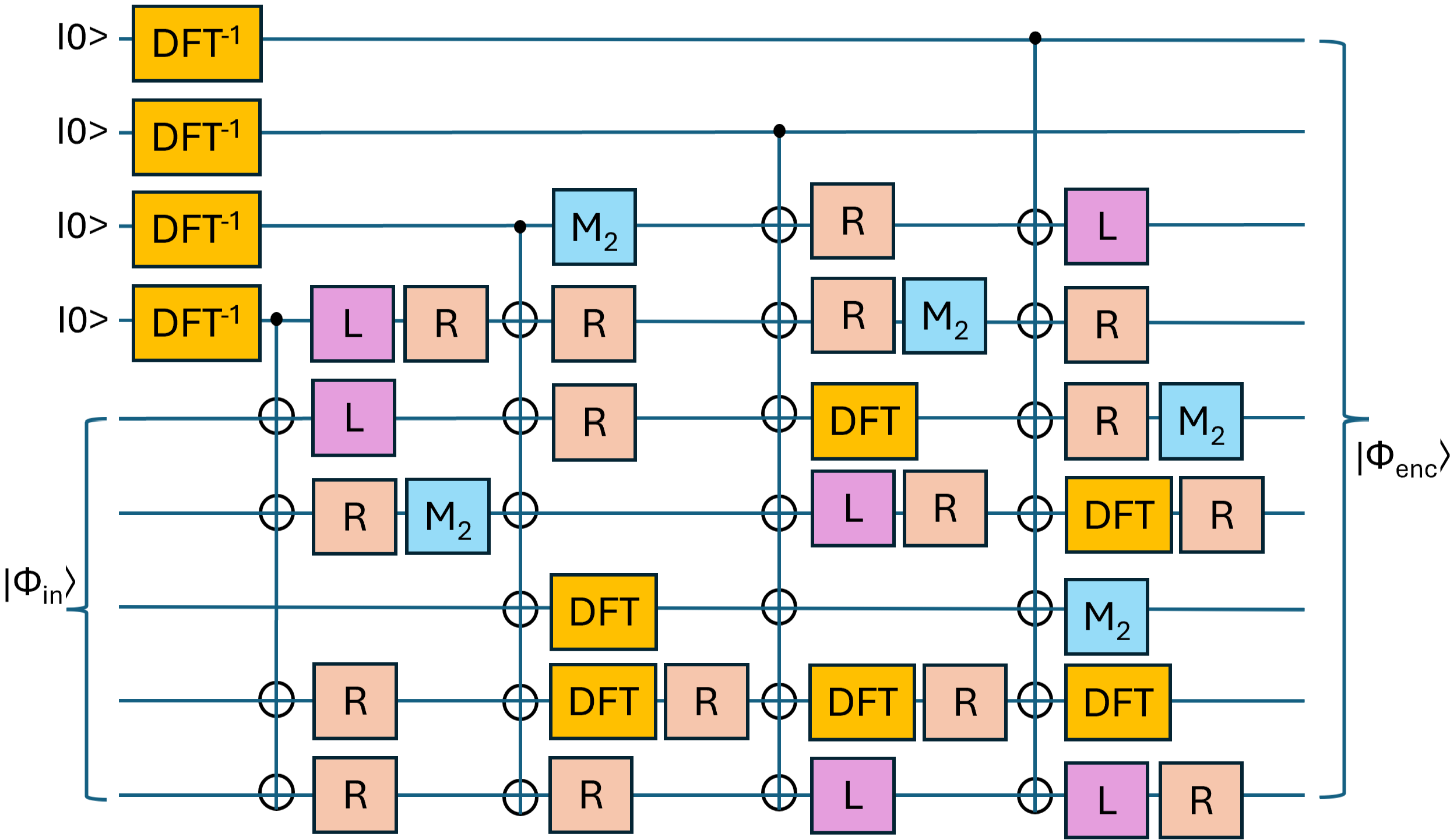}
  \caption{Encoding Circuit for Nine Qutrit Code \(\llbracket 9,5,3\rrbracket_3\) using the proposed gate set.}
  \label{fig:circuit1}
\end{figure}

\begin{table}[H]
\caption{Table of operators for encoding circuit of \(\llbracket 5,1,3\rrbracket_3\) code}
\centering
\scalebox{1.1}{%
    \begin{tabular}{|c|c|}
    \hline
    \textbf{Stage} & \textbf{Operator Form} \\
    \hline
    $T_1$       & $\mathrm{id} \;\otimes\;\mathrm{DFT} \;\otimes\;L\;\otimes\;M_{2} \;\otimes\;\mathrm{id}$ \\
    \hline
    $A1$ & $\mathrm{ADD}^{(1,2)}\;\mathrm{ADD}^{(1,3)}\;\mathrm{ADD}^{(1,4)}\;$ \\
    \hline
    $T2$           & $\mathrm{id} \;\otimes\;L \;\otimes\;M_2\;\otimes\;\mathrm{DFT} \;\otimes\;
M_2$ \\
    \hline
    $A_2$           & $\mathrm{ADD}^{(2,3)}\;\mathrm{ADD}^{(2,4)}\;\mathrm{ADD}^{(2,5)}\;$ \\
    \hline
    $T_3$          & $M_2 \;\otimes\;\mathrm{id} \;\otimes\;\mathrm{DFT}R\;\otimes\;\mathrm{DFT}R\;\otimes\;
\mathrm{DFT}$ \\
    \hline
    $A_3$       & $\mathrm{ADD}^{(3,1)}\;\mathrm{ADD}^{(3,4)}\;\mathrm{ADD}^{(3,5)}\;$ \\
    \hline
    $T_4$       & $\mathrm{id}\;\otimes\;M_2 \;\otimes\;M_2 \;\otimes\;\mathrm{id}\;\otimes\;L$ \\
    \hline
    $A_4$       & $\mathrm{SWAP}^{(4,5)}\;\mathrm{ADD}^{(4,1)}\;\mathrm{ADD}^{(4,2)}\;\mathrm{ADD}^{(4,3)}\;$ \\
    \hline
    \end{tabular}
}

\label{tab:qutrit_transformations_scalebox}
\end{table}

\begin{table}[H]
\caption{Table of operators for encoding circuit of \(\llbracket 9,5,3\rrbracket_3\) code}
\centering
\begin{tabular}{|c|c|}
\hline
\textbf{Stage} & \textbf{Operator Form} \\
\hline
$T_1$ & \begin{tabular}{@{}c@{}}$\mathrm{id} \otimes \mathrm{id} \otimes L \otimes R \otimes RM_{2} \otimes$ \\ $\mathrm{DFTR} \otimes M_{2} \otimes \mathrm{DFT} \otimes LR$\end{tabular} \\
\hline
$A_1$ & \begin{tabular}{@{}c@{}}$\mathrm{ADD}(1,3)\;\mathrm{ADD}(1,4)\;\mathrm{ADD}(1,5)\;\mathrm{ADD}(1,6)\;$ \\ $\mathrm{ADD}(1,7)\;\mathrm{ADD}(1,8)\;\mathrm{ADD}(1,9)$\end{tabular} \\
\hline
$T_2$ & \begin{tabular}{@{}c@{}}$\mathrm{id} \otimes \mathrm{id} \otimes R \otimes RM_{2} \otimes \mathrm{DFT} \otimes$ \\ $LR \otimes \mathrm{id} \otimes \mathrm{DFTR} \otimes L$\end{tabular} \\
\hline
$A_2$ & \begin{tabular}{@{}c@{}}$\mathrm{ADD}(2,3)\;\mathrm{ADD}(2,4)\;\mathrm{ADD}(2,5)\;\mathrm{ADD}(2,6)\;$ \\ $\mathrm{ADD}(2,7)\;\mathrm{ADD}(2,8)\;\mathrm{ADD}(2,9)$\end{tabular} \\
\hline
$T_3$ & \begin{tabular}{@{}c@{}}$\mathrm{id} \otimes \mathrm{id} \otimes M_{2} \otimes R \otimes R \otimes$ \\ $\mathrm{id} \otimes \mathrm{DFT} \otimes \mathrm{DFTR} \otimes R$\end{tabular} \\
\hline
$A_3$ & \begin{tabular}{@{}c@{}}$\mathrm{ADD}(3,4)\;\mathrm{ADD}(3,5)\;\mathrm{ADD}(3,6)\;$ \\ $\mathrm{ADD}(3,7)\;\mathrm{ADD}(3,8)\;\mathrm{ADD}(3,9)$\end{tabular} \\
\hline
$T_4$ & \begin{tabular}{@{}c@{}}$\mathrm{id} \otimes \mathrm{id} \otimes \mathrm{id} \otimes LR \otimes L \otimes$ \\ $RM_{2} \otimes \mathrm{id} \otimes R \otimes R$\end{tabular} \\
\hline
$A_4$ & \begin{tabular}{@{}c@{}}$\mathrm{ADD}(4,5)\;\mathrm{ADD}(4,6)\;$ \\ $\mathrm{ADD}(4,8)\;\mathrm{ADD}(4,9)$\end{tabular} \\
\hline
\end{tabular}
\label{tab:my_new_table}
\end{table}

\begin{table}[H]
\caption{Table of operators for encoding circuit of \(\llbracket 10,6,3\rrbracket_5\) code}
\centering
\scalebox{1.1}{
\begin{tabular}{|c|c|}
\hline
\textbf{Stage} & \textbf{Operator Form} \\
\hline
$T_1$ & \begin{tabular}{@{}c@{}}$\mathrm{DFT} \otimes Q\mathrm{DFT}Q \otimes PS \otimes P\mathrm{DFT} \otimes S \otimes \mathrm{DFT}$ \\ $\otimes Q\mathrm{DFT}Q \otimes PS \otimes P\mathrm{DFT} \otimes S$\end{tabular} \\
\hline
$A_1$ & \begin{tabular}{@{}c@{}}$\mathrm{ADD}(1,2)\;\mathrm{ADD}(1,3)\;\mathrm{ADD}(1,4)$\\
$\;\mathrm{ADD}(1,5)\;\mathrm{ADD}(1,6)$ \\ $\;\mathrm{ADD}(1,7)\;\mathrm{ADD}(1,8)\;\mathrm{ADD}(1,9)\;\mathrm{ADD}(1,10)$\end{tabular} \\
\hline
$T_2$ & \begin{tabular}{@{}c@{}}$\mathrm{id} \otimes \mathrm{DFT}PS \otimes \mathrm{DFT}P \otimes Q \otimes SQ \otimes \mathrm{DFT}$ \\ $\otimes \mathrm{DFT}PS \otimes \mathrm{DFT}P \otimes Q \otimes SQ$\end{tabular} \\
\hline
$A_2$ & \begin{tabular}{@{}c@{}}$\mathrm{ADD}(2,3)\;\mathrm{ADD}(2,4)\;\mathrm{ADD}(2,5)\;\mathrm{ADD}(2,6)$ \\ $\;\mathrm{ADD}(2,7)\;\mathrm{ADD}(2,8)\;\mathrm{ADD}(2,9)\;\mathrm{ADD}(2,10)$\end{tabular} \\
\hline
$T_3$ & \begin{tabular}{@{}c@{}}$\mathrm{id} \otimes \mathrm{DFT}\mathrm{DFT} \otimes Q\mathrm{DFT}Q \otimes S \otimes S\mathrm{DFT} \otimes \mathrm{id}$ \\ $\otimes SS \otimes Q\mathrm{DFT}Q \otimes S \otimes S\mathrm{DFT}$\end{tabular} \\
\hline
$A_3$ & \begin{tabular}{@{}c@{}}$\mathrm{ADD}(3,1)\;\mathrm{ADD}(3,2)\;\mathrm{ADD}(3,4)$\\
$\;\mathrm{ADD}(3,5)\;\mathrm{ADD}(3,6)$ \\ $\;\mathrm{ADD}(3,7)\;\mathrm{ADD}(3,8)\;\mathrm{ADD}(3,9)\;\mathrm{ADD}(3,10)$\end{tabular} \\
\hline
$T_4$ & \begin{tabular}{@{}c@{}}$\mathrm{id} \otimes \mathrm{DFT}\mathrm{DFT} \otimes \mathrm{id} \otimes PS \otimes \mathrm{DFT}Q$ \\ $\otimes \mathrm{id} \otimes \mathrm{id} \otimes PP \otimes PS \otimes \mathrm{DFT}Q$\end{tabular} \\
\hline
$A_4$ & \begin{tabular}{@{}c@{}}$\mathrm{ADD}(4,2)\;\mathrm{ADD}(4,5)\;\mathrm{ADD}(4,7)\;\mathrm{ADD}$\\
$(4,8)\;\mathrm{ADD}(4,9)\;\mathrm{ADD}(4,10)$\end{tabular} \\
\hline
\end{tabular}
} 
\label{tab:my_new_table}
\end{table}

\begin{table}[H]
\caption{Symplectic representation of gate sets for different values of d and set\_size}
\centering
\scalebox{0.75}{
\begin{tabular}{|c|c|c|c|}
\hline
\textbf{d} & \textbf{\begin{tabular}{@{}c@{}}set\_size\end{tabular}} & \textbf{\begin{tabular}{@{}c@{}}Symp. Rep. of Gates \\ (Section III)\end{tabular}} & \textbf{\begin{tabular}{@{}c@{}}Symp. Rep. of Gates \\ (Proposed)\end{tabular}} \\
\hline

\multirow{2}{*}{3} & 3 & 
\begin{tabular}{@{}c@{}}$\text{DFT} = \begin{pmatrix} 0 & 2 \\ 1 & 0 \end{pmatrix}, P_1 = \begin{pmatrix} 1 & 1 \\ 0 & 1 \end{pmatrix},$ \\ $P_2 = \begin{pmatrix} 1 & 2 \\ 0 & 1 \end{pmatrix}$\end{tabular} & 
\begin{tabular}{@{}c@{}}$\text{DFT} = \begin{pmatrix} 0 & 2 \\ 1 & 0 \end{pmatrix}, A_1 = \begin{pmatrix} 1 & 2 \\ 2 & 2 \end{pmatrix},$ \\ $A_2 = \begin{pmatrix} 2 & 1 \\ 0 & 2 \end{pmatrix}$\end{tabular} \\
\cline{2-4}
& 4 & 
\begin{tabular}{@{}c@{}}$\text{DFT} = \begin{pmatrix} 0 & 2 \\ 1 & 0 \end{pmatrix}, P_1 = \begin{pmatrix} 1 & 1 \\ 0 & 1 \end{pmatrix},$ \\ $P_2 = \begin{pmatrix} 1 & 2 \\ 0 & 1 \end{pmatrix}, M_2 = \begin{pmatrix} 2 & 0 \\ 0 & 2 \end{pmatrix}$\end{tabular} & 
\begin{tabular}{@{}c@{}}$\text{DFT} = \begin{pmatrix} 0 & 2 \\ 1 & 0 \end{pmatrix}, B_1 = \begin{pmatrix} 0 & 1 \\ 2 & 0 \end{pmatrix},$ \\ $B_2 = \begin{pmatrix} 2 & 0 \\ 0 & 2 \end{pmatrix}, B_3 = \begin{pmatrix} 0 & 2 \\ 1 & 2 \end{pmatrix}$\end{tabular} \\
\hline

\multirow{3}{*}{5} & 3 & 
\begin{tabular}{@{}c@{}}$\text{DFT} = \begin{pmatrix} 0 & 4 \\ 1 & 0 \end{pmatrix}, P_1 = \begin{pmatrix} 1 & 1 \\ 0 & 1 \end{pmatrix},$ \\ $P_2 = \begin{pmatrix} 1 & 2 \\ 0 & 1 \end{pmatrix}$\end{tabular} & 
\begin{tabular}{@{}c@{}}$\text{DFT} = \begin{pmatrix} 0 & 4 \\ 1 & 0 \end{pmatrix}, C_1 = \begin{pmatrix} 3 & 0 \\ 4 & 2 \end{pmatrix},$ \\ $C_2 = \begin{pmatrix} 1 & 4 \\ 3 & 3 \end{pmatrix}$\end{tabular} \\
\cline{2-4}
& 4 & 
\begin{tabular}{@{}c@{}}$\text{DFT} = \begin{pmatrix} 0 & 4 \\ 1 & 0 \end{pmatrix}, P_1 = \begin{pmatrix} 1 & 1 \\ 0 & 1 \end{pmatrix},$ \\ $P_2 = \begin{pmatrix} 1 & 2 \\ 0 & 1 \end{pmatrix}, M_2 = \begin{pmatrix} 3 & 0 \\ 0 & 2 \end{pmatrix}$\end{tabular} & 
\begin{tabular}{@{}c@{}}$\text{DFT} = \begin{pmatrix} 0 & 4 \\ 1 & 0 \end{pmatrix}, D_1 = \begin{pmatrix} 2 & 3 \\ 2 & 1 \end{pmatrix},$ \\ $D_2 = \begin{pmatrix} 2 & 0 \\ 0 & 3 \end{pmatrix}, D_3 = \begin{pmatrix} 4 & 4 \\ 0 & 4 \end{pmatrix}$\end{tabular} \\
\cline{2-4}
& 5 & 
\begin{tabular}{@{}c@{}}$\text{DFT} = \begin{pmatrix} 0 & 4 \\ 1 & 0 \end{pmatrix}, P_1 = \begin{pmatrix} 1 & 1 \\ 0 & 1 \end{pmatrix},$ \\ $P_2 = \begin{pmatrix} 1 & 2 \\ 0 & 1 \end{pmatrix}, M_2 = \begin{pmatrix} 3 & 0 \\ 0 & 2 \end{pmatrix},$ \\ $M_3 = \begin{pmatrix} 2 & 0 \\ 0 & 2 \end{pmatrix}$\end{tabular} & 
\begin{tabular}{@{}c@{}}$\text{DFT} = \begin{pmatrix} 0 & 4 \\ 1 & 0 \end{pmatrix}, E_1 = \begin{pmatrix} 0 & 3 \\ 3 & 1 \end{pmatrix},$ \\ $E_2 = \begin{pmatrix} 2 & 4 \\ 2 & 2 \end{pmatrix}, E_3 = \begin{pmatrix} 4 & 0 \\ 1 & 4 \end{pmatrix},$ \\ $E_4 = \begin{pmatrix} 0 & 4 \\ 1 & 4 \end{pmatrix}$\end{tabular} \\
\hline
\end{tabular}
} 
\label{tab:my_new_table}
\end{table}

We also tabulate the improvement that our algorithm proposes for various \texttt{set\_size} for different codes in Table 7.

\newcolumntype{C}[1]{>{\centering\arraybackslash}p{#1}}
\begin{table}[ht]
\caption{Comparison of single qudit gate count and circuit depth for different quantum codes and generating gate sets}
\centering
\scalebox{1}{
\begin{tabular}{| C{1.2cm} | C{1.1cm} | C{1.1cm} | C{1.1cm} | C{0.65cm} | C{0.7cm} |}
\hline
\textbf{Codes} & \textbf{\begin{tabular}{@{}c@{}}Total\\ \# of Gates \\ in the \\ genera-\\ting set \\(set\_size)\end{tabular}} & \textbf{\begin{tabular}{@{}c@{}}Total \# of \\Single\\ Qudit\\Gates \\ (Sec. III)\end{tabular}} & \textbf{\begin{tabular}{@{}c@{}}Total \# of \\Single\\ Qudit\\Gates \\ (Proposed)\end{tabular}} & \textbf{\begin{tabular}{@{}c@{}}Gate \\ Redu-\\ction\end{tabular}} & \textbf{\begin{tabular}{@{}c@{}}Depth \\ Redu-\\ction\end{tabular}} \\
\hline

\multirow{2}{*}{$[[5,1,3]]_3$} & 3 & 32 & 18 & 44 \% & 42 \% \\
\cline{2-6}
& 4 & 19 & 16 & 16 \% & 21 \% \\
\hline

\multirow{2}{*}{$[[7,1,3]]_3$} & 3 & 16 & 14 & 13 \% & 17 \% \\
\cline{2-6}
& 4 & 10 & 8 & 20 \% & 29 \%\\
\hline

\multirow{2}{*}{$[[9,5,3]]_3$} & 3 & 51 & 40 & 22 \% & 33 \%  \\
\cline{2-6}
& 4 & 39 & 33 & 15 \% & 0 \%\\
\hline

\multirow{3}{*}{$[[10,6,3]]_5$} & 3 & 84 & 66 & 21 \% & 14 \%\\
\cline{2-6}
& 4 & 62 & 57 & 9 \% & 12 \% \\
\cline{2-6}
& 5 & 54 & 49 & 9 \% & 20 \% \\
\hline

\end{tabular}
} 
\label{tab:gate_count_comparison}
\end{table}

\section{Proposed Operators for d = 3}
Applying the framework to derive the optimal generating set of operators for the qutrit case we found $\{{L}, {\text{DFT}}, {M_2}, {R}\}$ to be the generating set of operators. Two of the operators $M_2$ and $DFT$ have been defined in section II-G. Now we give two theorems that will completely specify the other two operators, i.e., L and R.\\

\begin{theorem}
The qutrit Clifford operator $L$,
\begin{equation}
    L = \frac{1}{\sqrt{3}} \sum_{j,k=0}^{2} \omega^{2jk} |j\rangle \langle k|
\end{equation}
where $\omega = e^{2\pi i/3}$, realizes the symplectic transformation in $\mathbb{F}_3^2$ represented by the matrix:
\begin{equation}
    \overline{L} = 
    \begin{pmatrix}
        0 & 1 \\
        2 & 0
    \end{pmatrix} .
\end{equation}
This corresponds to the mapping of Pauli operators $X(\alpha) \mapsto Z(\alpha)$ and $Z(\beta) \mapsto X(2\beta)$.
\end{theorem}

\begin{proof}
To prove this, we demonstrate the conjugation of the generalized Pauli operators $X(\alpha)$ and $Z(\beta)$ by the operator $L$. The resulting transformation of the exponent vector $(\alpha, \beta)$ must follow the rule $(\alpha', \beta') = (\alpha, \beta) \cdot \overline{L}$.

\subsubsection*{1. Action on $X(\alpha)$}
We compute the action of the conjugated operator on an arbitrary basis state $|k\rangle$:
\begin{align}
    (L^{-1} X(\alpha) L) |k\rangle &= L^{-1} X(\alpha) \left(\frac{1}{\sqrt{3}} \sum_{j=0}^{2} \omega^{2kj} |j\rangle\right) \\
    &= \frac{1}{\sqrt{3}} L^{-1} \left(\sum_{j=0}^{2} \omega^{2kj} |j+\alpha\rangle\right)
\end{align}
Let $p = j+\alpha$, which implies $j = p-\alpha$. Substituting this into the sum:
\begin{align}
    &= \frac{1}{\sqrt{3}} L^{-1} \left(\sum_{p=0}^{2} \omega^{2k(p-\alpha)} |p\rangle\right) \\
    &= \frac{1}{3} \sum_{p=0}^{2} \omega^{2kp-2k\alpha} \left(\sum_{m=0}^{2} \omega^{-pm} |2m\rangle\right) \\
    &= \frac{1}{3} \omega^{-2k\alpha} \sum_{m=0}^{2} \left(\sum_{p=0}^{2} \omega^{p(2k-m)}\right) |2m\rangle
\end{align}
The inner sum over $p$ is non-zero only if $2k-m \equiv 0 \pmod 3$, which implies $m=2k$.
\begin{align}
    &= \omega^{-2k\alpha} |2(2k)\rangle = \omega^{-2k\alpha} |4k\rangle = \omega^{k\alpha} |k\rangle
\end{align}
This is the action of $Z(\alpha)$. Thus, $L^{-1} X(\alpha) L = Z(\alpha)$. The exponent vector $(1,0)$ for $X(\alpha)$ maps to $(0,1)$ for $Z(\alpha)$, confirming the first row of $\overline{L}$ is $(0,1)$.

\subsubsection*{2. Action on $Z(\beta)$}
A similar calculation shows the transformation of $Z(\beta)$:
\begin{align}
    (L^{-1} Z(\beta) L) |k\rangle &= L^{-1} Z(\beta) \left(\frac{1}{\sqrt{3}} \sum_{j=0}^{2} \omega^{2kj} |j\rangle\right) \\
    &= \frac{1}{\sqrt{3}} L^{-1} \left(\sum_{j=0}^{2} \omega^{2kj} \omega^{\beta j} |j\rangle\right) \\
    &= \frac{1}{3} \sum_{j=0}^{2} \omega^{(2k+\beta)j} \left(\sum_{m=0}^{2} \omega^{-jm} |2m\rangle\right) \\
    &= \frac{1}{3} \sum_{m=0}^{2} \left(\sum_{j=0}^{2} \omega^{j(2k+\beta-m)}\right) |2m\rangle
\end{align}
The inner sum over $j$ is non-zero only if $2k+\beta-m \equiv 0 \pmod 3$, which implies $m=2k+\beta$.
\begin{align}
    &= |2(2k+\beta)\rangle = |4k+2\beta\rangle = |k+2\beta\rangle
\end{align}
This is the action of $X(2\beta)$. Thus, $L^{-1} Z(\beta) L = X(2\beta)$. The exponent vector $(0,1)$ for $Z(\beta)$ maps to $(2,0)$ for $X(2\beta)$, confirming the second row of $\overline{L}$ is $(2,0)$.

These two transformations confirm that the operator $L$ realizes the symplectic transformation described by $\overline{L}$.\\
\end{proof}

\begin{theorem}
The qutrit Clifford operator, $R$, is defined by:
\begin{equation}
    R = \frac{1}{\sqrt{3}} \sum_{p,q=0}^{2} \omega^{2q^2-2pq} \ket{p}\bra{q}
\end{equation}
where $\omega = e^{2\pi i/3}$, realizes the symplectic transformation in $\mathbb{F}_3^2$ represented by the matrix:
\begin{equation}
    \overline{R} = 
    \begin{pmatrix}
        0 & 2 \\
        1 & 2
    \end{pmatrix}.
\end{equation}
This corresponds to the mapping of Pauli operators $X(\alpha) \mapsto Z(2\alpha)$ and $Z(\beta) \mapsto X(\beta)Z(2\beta)$, up to a phase factor.
\end{theorem}

\begin{proof}
The proof is established by demonstrating how the operator $R$ conjugates the generalized Pauli operators. We compute the action of the transformed operators on an arbitrary basis state $\ket{k}$ using the definitions of the operators. The transformation of the exponent vector $(\alpha, \beta)$ follows the rule $(\alpha', \beta') = (\alpha, \beta) \cdot \overline{R}$.

\subsubsection*{1. Action on $X(\alpha)$}
We compute the action of the conjugated operator on $\ket{k}$ by applying each operator in sequence.

\begin{align}
    \MoveEqLeft (R^{-1}X(\alpha)R)\ket{k} \notag \\
    &= \begin{multlined}[t]
        R^{-1}X(\alpha)
        \left( \frac{1}{\sqrt{3}} \sum_{j=0}^2 \omega^{2k^2-2jk} \ket{j} \right)
       \end{multlined} \\
    &= \frac{1}{\sqrt{3}} \sum_{j=0}^2 \omega^{2k^2-2jk} R^{-1}X(\alpha)\ket{j} \\
    &= \frac{1}{\sqrt{3}} \sum_{j=0}^2 \omega^{2k^2-2jk} R^{-1}\ket{j+\alpha} \\
    &= \begin{multlined}[t]
        \frac{1}{3} \sum_{j=0}^2 \omega^{2k^2-2jk}
        \left( \sum_{m=0}^2 \omega^{-2m^2+2m(j+\alpha)} \ket{m} \right)
       \end{multlined}
\end{align}
Now, we reorder the sums and group the exponents of $\omega$:
\begin{align}
    &= \frac{1}{3} \sum_{m=0}^2 \omega^{2k^2-2m^2+2m\alpha} \left( \sum_{j=0}^2 \omega^{j(-2k+2m)} \right) \ket{m}
\end{align}
The inner sum over $j$ is non-zero only if $-2k+2m \equiv 0 \pmod 3$, which simplifies to $m=k$. This sum evaluates to $3\,\delta_{m,k}$.
\begin{align}
    &= \frac{1}{3} \omega^{2k^2-2k^2+2k\alpha} (3\,\delta_{m,k}) \ket{k} \\
    &= \omega^{2k\alpha}\ket{k} = Z(2\alpha)\ket{k}
\end{align}
Since this holds for any $\ket{k}$, we find $R^{-1}X(\alpha)R = Z(2\alpha)$. The vector $(1,0)$ for $X(\alpha)$ maps to $(0,2)$ for $Z(2\alpha)$, confirming the first row of $\overline{R}$ is $(0,2)$.

\subsubsection*{2. Action on $Z(\beta)$}
Similarly, we compute the transformation of $Z(\beta)$:
\begin{align}
    \MoveEqLeft (R^{-1}Z(\beta)R)\ket{k} \notag \\
    &= \begin{multlined}[t]
        R^{-1}Z(\beta)
        \left( \frac{1}{\sqrt{3}} \sum_{j=0}^2 \omega^{2k^2-2jk} \ket{j} \right)
       \end{multlined} \\
    &= \begin{multlined}[t]
        \frac{1}{\sqrt{3}} \sum_{j=0}^2 \omega^{2k^2-2jk}
        R^{-1} ( \omega^{\beta j}\ket{j} )
       \end{multlined} \\
    &= \begin{multlined}[t]
        \frac{1}{3} \sum_{j=0}^2 \omega^{2k^2-2jk+\beta j}
        \left( \sum_{m=0}^2 \omega^{-2m^2+2mj} \ket{m} \right)
       \end{multlined}
\end{align}
Reordering the sums and grouping exponents:
\begin{align}
    &= \frac{1}{3} \sum_{m=0}^2 \omega^{2k^2-2m^2} \left( \sum_{j=0}^2 \omega^{j(-2k+\beta+2m)} \right) \ket{m}
\end{align}
The inner sum over $j$ yields $3\,\delta_{-2k+\beta+2m, 0}$, which forces $2m = 2k-\beta$, or $m = k-2\beta \equiv k+\beta \pmod 3$.
\begin{align}
    &= \omega^{2k^2 - 2(k+\beta)^2} \ket{k+\beta} \\
    &= \omega^{2k^2 - 2(k^2+2k\beta+\beta^2)} \ket{k+\beta} \\
    &= \omega^{-4k\beta - 2\beta^2} \ket{k+\beta} = \omega^{2k\beta - 2\beta^2} \ket{k+\beta}
\end{align}
This final expression corresponds to the action of $\omega^{-2\beta^2} X(\beta) Z(2\beta)$ on $\ket{k}$. 

Thus, $R^{-1}Z(\beta)R = \omega^{-2\beta^2} X(\beta)Z(2\beta)$. The vector $(0,1)$ for $Z(\beta)$ maps to $(1,2)$ for $X(\beta)Z(2\beta)$, confirming the second row of $\overline{R}$ is $(1,2)$.

These two transformations confirm that the operator $R$ realizes the symplectic transformation described by $\overline{R}$.
\end{proof}

\section{Conclusion}
This work introduced a systematic framework for synthesizing encoder circuits for prime-dimension stabilizer codes by optimizing the underlying generating gate sets. By casting encoder synthesis as shortest-path search on the Cayley graph of $SL(2,\mathbb{F}_d)$ induced by candidate generators, the method delivers encoder implementations that minimize single-qudit gate count and reduce circuit depth while preserving universality over the Clifford group. For qutrit codes, we demonstrated reductions of up to $44\%$ in the single qudit gate count and $42\%$ in depth for the $[[5,1,3]]_3$ code, and $22\%$/$33\%$ for the $[[9,5,3]]_3$ code; for the ququint $[[10,6,3]]_5$ code, we achieved reductions up to $21\%$ in the single qudit gate count and $20\%$ in depth across the evaluated generating gate set sizes. Beyond these empirical gains, we provided explicit operator constructions (e.g., $L$ and $R$) and proofs establishing the isomorphism between the action of our generators and $SL(2,\mathbb{F}_3)$/$SL(2,\mathbb{F}_5)$, enabling direct compilation to efficient Clifford sequences.

The practical implications are twofold. First, the reduced depth alleviates coherence-time pressure in near-term qudit platforms, improving the viability of non-binary QEC encoders. Second, the framework is compiler-friendly: once hardware-native primitives are mapped to their symplectic representations, the same search-and-score machinery can target alternative cost models without altering code semantics.

This study also opens several avenues for future work. A hardware-aware cost model that accounts for asymmetric native gate times, calibration overheads, and crosstalk can further sharpen depth and fidelity. Extending the framework beyond prime $d$ to prime-power dimensions and constrained native gate sets will broaden applicability. Finally, integrating additional fault-tolerant constraints (e.g., transversality requirements and magic-state overheads), exploring heuristic/ILP-guided search to complement exhaustive enumeration, and validating the synthesized encoders on superconducting, photonic, and ion-trap qudit platforms are natural next steps toward deployable non-binary QEC compilers.

{
\bibliography{papers}}

\end{document}